\documentclass[a4paper,11pt]{article}
\usepackage{amsmath}
\usepackage{amssymb}
\usepackage{amsthm}
\usepackage{mathrsfs}
\numberwithin{equation}{section}
\theoremstyle{remark}
\newtheorem{remark}{Remark}
\theoremstyle{definition}
\newtheorem{ngauss}{Definition}
\theoremstyle{plain}
\newtheorem{propos}{Proposition}
\theoremstyle{plain}
\newtheorem{theorem}{Theorem}
\date{}
\author{Radhakrishnan Balu}
\title{Quantum filter processes driven by Markovian white noises have classical versions}
\begin{document}
\maketitle
\begin{abstract}
We study quantum filters that are driven by basic quantum noises and construct classical versions. Our approach is based on exploiting the  quantum markovian component of the observation and measurement processes of the filters. This approach leads in a natural way the classical versions for a class of quantum filters. We consider quantum white noises derived from Wiener and Poisson processes that drive the signal and measurement processes and derive the recursive filtering equations using classical machinery. 
\end{abstract}
\section{Introduction}

Interest in quantum stochastic filters is fueled by experiments in quantum feedback
control \cite{Ar} techniques. There are also theoretical treatments of such processes and one such approach is based on the framework of Hudson-Parthasarathy version of quantum stochastic calculus \cite{Pa}.  For an excellent introduction to quantum filters and the theoretical modeling based on quantum stochastic calculus the reader is referred to the refernce \cite{Lu}. The authors have derived the recursive relations for the quantum stochastic filters driven by quantum white noises for the case of homodyne detection and photon counting experiments. Even though the underlying noises are quantum in nature the resulting filters are classical stochastic processes. Intrigued by this result we set to explore the possibility of deriving the filtering equations within the framework of classical stochastic methods. Our hope is that the rich literature available in classical stochastic processes can be exploited to gather further insights into quatum stochastic filters. As our aim is to establish the applicability of the classical techniques we will focus on deriving the equations presented in \cite{Lu} and work with the physical model and the notations used there in. For the classical techniques we will adapt the methods presented in \cite{AB}, which is a freely downloadable book on the introduction to stochastic calculus and filetering, to derive the equations for homodyne detection. For the case of photon counting problem we will use the techniques from the reference \cite{Pl}.

Our work is motivated by the research that relates classical and quantum stochastic processes and the influence each field has on the other. The relation between non-commutative and classical probabilities has been explored since 1970 \cite{FB,BKS} and continues to grow at a steady pace \cite{WB}. Quantum probability which is a non-commutative probability defined on a von Neumann alegbra \cite{PA}. It is well known that we can construct classical versions from a quantum probability space using spectral theorem if the underlying *-algebra is commutative \cite{Lu}. There are ways to realize the classical versions even if the underlying algebra is a non-commutative von Neumann algebra. As has been established in \cite{BKS} a restricted class of quantum markovian processes can be realized as classical versions. Our report is organized as follows: first we describe the white noise processes that make the physical model in detail and establish them as quantum markovian in the sense of \cite{BKS}. Then we build the corresponding classical versions and prove that they have the same covariance as their quantum counterparts by adapting the techniques from \cite{FB}. Finally we derive Martingales on a classical probability space for the measurement processes and express them as integral of some Brownian motions. This will establish the classical counter part for our quantum filtering problem. \cite{AB} 

\section {Quantum stochastic calculus and filtering}
We start with the physical model described in \cite{Lu} for the quantum filtering problem that is applicable to quantum optics. An example of such a system would be an atom placed in an electromagnetic field. The model consists of an intial system defined on a finite dimensional Hilbert space $\mathscr{H}_0$ that interacts with an environment described by a symmetric boson fock space $\Gamma = \Gamma_s(L^2([0,T))$. Let us work with the quantum probablity space $(\mathcal{A},\mathbb{P})$ where $\mathcal{A}$ is a von Neumann algebra defined as $B(\mathscr{H}_0)\otimes\mathscr{W}$ where $\mathscr{W} = B(\Gamma)$, the von Neumann algebra generated by the bounded operators on $\Gamma$. We let $\mathbb{P}$ as a normal state defined as $\rho\otimes\phi$. In this work we will assume the state $\phi$ to be the fock vacuum state and all the operators to be bounded. Thus the state of our von Neumann algebra is faithful normal.

Accordingly, the quantum stochastic differential equations describing the filetering problem are
\begin{equation}
dj_t(X) = j_t(\mathcal{L}_{L,H}(X))dt + j_t([L^*,X])dA_t + j_t([X,L])dA_t; \label{sg}
\end{equation}
\begin{equation}
dY_t^W = j_t(L+L^*)dt + dA_t + dA^*_t;  \label{hd}
\end{equation}
\begin{equation}
dY_t^\Lambda = d\Lambda_t + j_t(L)dA^*_t + j_t(L^*)dA_t + j_t(L^*L)dt; \label{pc}
\end{equation}
where X is an obeservable that belongs to the bounded set of operators of $\mathscr{H}_0$ and $j_t(X) = U_t^*(X\otimes)U_t$ describes the evolution of the observable in Heisenberg picture. The unitary operator $U_t$ satisfies the quantum stochastic differential equation
\begin{equation}
dU_t = \{LdA^*_t - L^*dA_t -\frac{1}{2}L^*Ldt -iHdt\}dU_t,    U_0 = I
\end{equation}

In the above, equation $\eqref{sg}$ represents the signal process, $\eqref{hd}$ represents the homodyne detection, and $\eqref{pc}$ represents the photon counting measurement prcoesses.
The filetering problem is stated as the derivation of conditional law $\pi_t(X) = \mathbb{P}(j_t(X)|\mathscr{Y}_t)$ where $\mathscr{Y}_t$ is the von Neumann algebra generated by the observation process which is commutative and thus a classical process.

Let us consider the processes, $Z_t(f) = \alpha_tA_t(f) + \beta_tA^*_t(f)$ and $W_t(f) = Z_t(f) + \gamma_t\Lambda_t(f)$ where $A_t$, $A^*_t$, and $\Lambda_t$ are the creation, annihilation, and conservation or gauge processes of the fock space, $\alpha_t$, $\beta_t$, and $\gamma_t$ are time dependent complex variables, and $f\in \Gamma$ . Let us call them as $\mathscr{N}$-Gaussian (Non-commutative) and $\mathscr{N}$-Martingale processes. These are non-commutative processes due to the arbitrary time dependent coefficients and we will see the conditions under which they posses classical versions. These processes satisfy the differential equations
\begin{equation}
dZ_t = \alpha_tdA_t + \beta_tdA^*_t       \label{dg}
\end{equation}
\begin{equation}
dW_t = dZ_t + \gamma_td\Lambda_t           \label{dm}
\end{equation}

Let us take a closer look at the structure of the Fock space operators. In the context of quantum optics they are defined as $A_{(g)}$ and $A^*_{(g)}$ where g is the charecterestic function $1_{[0,t]}$ satisfying CCR on the exponential domain, and on the finite particle vectors with covariance \cite{Pa}
\begin{equation}
[A_g(t),A^*_g(s)] = min(s,t)
\end{equation}
The gauge process which is essentially the Poisson process is defined as 
\begin{equation}
(\Lambda_tf)(\tau) = N_t(\tau)f(t) = |\tau \cap [0,t]|f(\tau), \text{     } f\in\Gamma, \tau \in \Omega, t\in[0,T]
\end{equation}
It is clear that equations $\eqref{dg} and \eqref{dm}$ form the components of the signal and measurement processes. Following \cite{BKS} we can now define the $\mathscr{N}$-analogues of classical Gaussian processes based on their covariance. 
\begin{ngauss} The $\mathscr{N}$-Gaussian process $X^{nBM}_t$ with covariance min(s,t) is called a $\mathscr{N}$-Brownian motion.\\ \\
\end{ngauss}
As we mentioned earlier we adopt the notion of quantum markovian, quantum martingale, and independence defined in \cite{BKS}. 
\begin{ngauss} Let $(\mathcal{A},\mathbb{P})$ be a quantum probability space and $(X)_{t\in{T}}$ is a stochastic process on  $(\mathcal{A},\mathbb{P})$. Let
\begin{equation}
\mathscr{A}_{t]} := vN(X_u|u \leq t) \subset \mathscr{A} 
\end{equation}
\begin{equation}
\mathscr{A}_{[t]} := vN(X_t) \subset \mathscr{A}
\end{equation}
The process $X_{t\in{T}}$ is called quantum markovian if the following holds:
$\forall s,t \in T$ with $s \leq t$  \\
\begin{equation}
\mathbb{P}[X|\mathcal{A}_{s]}] \subset \mathcal{A}_{[s]}, \hspace{1.5cm} \forall X \in \mathcal{A}_{[t]}
\end{equation}
The conditional expectation $E[.|\mathcal{A}_{s]}]$ is a map from $\mathcal{A}$ to $\mathcal{A}_{s]}$. In our context this becomes, at Hilbert space level, a projection from $\mathscr{H_0}\otimes\Gamma$ to $\mathscr{H_0}\otimes\Gamma_{s]}$ (see \cite{BKS} for details)  .
\end{ngauss}
\begin{ngauss}
Let $(\mathscr{A}, \mathbb{P})$ be a quantum probability space, a process $(X)_{t\in{T}}$ defined on it is a called a martingale if
\begin{equation}
\mathbb{P}[X_t|\mathscr{A}_{s]}] = X_s, \text{      } \forall s \leq t
\end{equation}
\end{ngauss}

The following proposition will be helpful in establishing the martingale property of $\mathscr{N}$-Gaussian processes.
\begin{propos} (Ref \cite{BKS}) A  $\mathscr{N}$-Gaussian process is a martingale if and only if $Pf_{s]} = f_s  \forall s \leq t$ which is the case iff c(s,t) = c(s,s) $\forall s \leq t$.
\end{propos}
\begin{propos} (Ref \cite{BKS}) $\mathscr{N}$-Brownian motion and $\mathscr{N}$-Martingale processes are quantum markovian as well as quantum martingale.
\end{propos}
\begin{proof}
A $\mathscr{N}$-Gaussian process with covariance c is quantum markovian iff $\forall$ s,u,t $\in T$, such that s $\leq u \leq t$, 
\begin{equation}
c(t,s)c(u,u) = c(t,u)c(u,s)
\end{equation}  
Thus $\mathscr{N}$-Brownian motion processes is quantum markovian.
See the proof of theorem 4.9 in \cite{Sim} for the validity of equation (10).\\
We have c(t,s) = min(s,t) = s = c(s,s), the martingale property is established.\\
The gauge process being the Poisson process the markovian and martingale properties of $\mathscr{N}$-Martingale follow immediately.
\end{proof}
Having established the quantum markovianity of the $\mathscr{N}$-Gaussian and $\mathscr{N}$-Martingale processes we will build the classical version that preserves the covariance. 
\begin{ngauss} Let $(X_t)_{t\in T}$ be a stochastic process on a quantum probability space $(\mathscr{A},\mathbb{P})$. We call a process $(V_t)_{t\in T}$ taking values in $\mathbb{R}$ and defined on a classical probability space $(\Omega,\mathscr{U},\mathscr{P})$ a classical version of $(X_t)_{t\in T}$ if all time ordered moments coincide. That is $\forall{n} \in \mathscr{N}$, $\forall{t_1,\dots,t_n} \in T$  with $t_1\leq\dots\leq t_n$, and all bounded Borel functions $h_1,\dots,h_n$ on $\mathscr{R}$ the following equality holds:
\begin{equation}
\mathbb{P}[h_1(X_{t_1})\dots h_n(X_{t_n})] = \int_\Omega h_1(V_{t_1}(\omega))\dots  h_n(V_{t_n}(\omega))d\mathscr{P}(\omega)
\end{equation}
\end{ngauss}
\begin{remark} Our interest in time ordered moments is due to its applicability in covariance which is needed to characterize Gaussian processes. From the results on time ordered moments we infer covraince on two different epochs are preserved irrespective of their order. 

Our classical probability space is the real line $\mathscr{R}$ and the sigma field $\mathscr{U}$ is the sigma field generated by the borel sets. We define the probabilty measure and the classical version $\tilde{V}_t$ as $\mathscr{P}(\tilde{V}_t \in {B}) = \mathbb{P}(1_{B}(X_t))$, where $B \in {\mathscr{U}}$.
Let us define the transition operators for the quantum markov processes that is the same in the classical context. Then we realize the classical version using the transition operator, $\grave{a}$ la Kolmogorov construction. This a widely used mathematical technique as can be seen from \cite{BKS}, \cite{AFL}, and \cite{Ra}.
\end{remark}
\begin{ngauss} Let $(X_t)_{t\in{T}}$ as defined above and is self-adjoint $\forall t\in T$. Let $spect(X_t)$ and $\nu_t$ are the spectrum and distribution of $X_t$ respectively. Let us define the space
\begin{equation}
L^{\infty}(X_t) = vN(X_t) = L^{\infty}(spect(X_t),\nu_t).
\end{equation}
and the transition operator 
\begin{equation}
\mathscr{K}_{s,t}:L^{\infty}(X_t) \longrightarrow L^{\infty}(X_s) \text{,  with   $s\leq t$}.
\end{equation}
defined as
\begin{equation}
\mathbb{P}[h(X_t)|\mathscr{A}_{s]}]=\mathbb{P}[h(X_t)|\mathscr{A}_{[s]}]=(\mathscr{K}_{s,t}h)(X_s).
\end{equation}
\end{ngauss}
\begin{theorem} The $\mathscr{N}$-Gaussian process $X^{nBM}_t$ has a classical version $V_t$. In particular the covraince is preserved for the $V_t$ Gaussian process and hence is a classical Brownian mtion.
\end{theorem}
\begin{proof}
Following the line of reasoning in \cite{BKS} theorem 4.4 we get
\begin{equation}
\mathbb{P}[h_1(X_{t_1}),\dots,h_n(X_{t_n})] = \mathbb{P}[(h_1.\mathscr{K}_{t_1,t_2}(h_2.\mathscr{K}_{t_2,t_3}(h_3.\dots)))(X_{t_1})]
\end{equation}
It follows that for arbitrary $t_1$ and $t_2$ we have
\begin{equation}
\mathbb{P}[h_1(X_{t_1}),h_2(X{t_2})]=\mathscr{P}[h_1(V_{t_1}),h_2(V_{t_2})] = c(t_1,t_2) = min(t_1,t_2)
\end{equation}
\\In essense we have replaced $X^{nBM}_t$ that has non-commutative components with a single random variable such that the first order moments coincide. The resulting stochastic process is commutative on a classical probability space.
\end{proof}
\section{The filtering problem}
The objective is to find a least mean square estmate of a system of observable $X \in \mathscr{B}$ at time t, given the observations $Y_t$ up to this time, that is to obtain a recursive relation
\begin{equation}
\pi_t(X) = \mathbb{P}(j_t(X)|\mathscr{Y}_t) \label{fe}
\end{equation}
where $\mathscr{Y}_t =vN\{Y_s:0\leq s \leq t \} $ is the commutative von Neumann algebra generated by the observation process. Let us now adopt the definition of conditional expectation from \cite{Lu}.
\begin{ngauss} (Conditional expectation). Let $(\mathscr{A},\mathbb{P})$ be a quantum probability space and let $\mathscr{N} \subset \mathscr{A}$ be a commutative von Neumann subalgebra. The conditional expectation $\mathbb{P}(.|\mathscr{N})$ is a map from $\mathscr{N}'$ onto $\mathscr{N}$ such that the following holds:
\begin{align}
\mathbb{P}(.|\mathscr{N}):\mathscr{N}' \rightarrow \mathscr{N},\nonumber \\
\mathbb{P}(\mathbb{P}(B|\mathscr{N})A) = \mathbb{P}(BA), \text{   } \forall A \in \mathscr{N}, B \in \mathscr{N}' 
\end{align}
Here $\mathscr{N}'$ is the commutant of $\mathscr{N}$.
\end{ngauss}
\begin{remark}
$\mathscr{N}$ is the commutative von Neumann algebra generated by the observation process. That is, an algebra of compatible observables. The domain of the conditional map has to be the commutant of $\mathscr{N}$ in order to make sure that compatible observables are present in the same realization. If in addition B = $B^*$ then B is also in the commutative algebra and hence in the classical probability space by spectral theorem. The self-nondemolition property is established easily as in \cite{Lu}.
\end{remark}
We can now define all the objects of $\eqref{sg}, \eqref{hd}, and \eqref{pc}$ to the corresponding classical objects 
and derive the filtering equations using the classical results as in \cite{AB} and for the photon counting problem based on \cite{Pl}.  Let us rewrite them using $\eqref{dm} and \eqref{dm}$ as
\begin{equation}
dj_t(X) = j_t(\mathcal{L}_{L,H}(X))dt + dZ_t; \label{sig}
\end{equation}
\begin{equation}
dY_t^W = j_t(L+L^*)dt + dZ_t;  \label{hed}
\end{equation}
\begin{equation}
dY_t^\Lambda = j_t(L^*L)dt + dW_t; \label{phc}
\end{equation}
\begin{theorem}(Kushner-Stratonowich, ref \cite{AB})
Let X and Y are the signal and observation processes satisfying the following classical stochastic differential equations
\begin{equation}
dX_t = f(t, X_t)dt + \sigma(t,X)dV_t
\end{equation}
\begin{equation}
dY_t = h(t, X_t)dt + dW_t, \text{     } Y_0 = 0
\end{equation}
where $V_t$ and $W_t$ are two independent Brownian motions.
Then the filtering equation is given by
\begin{equation}
d\pi_t(\phi) =  \pi_t(A_s(\phi))+ (\pi_t(h^T\phi)-\pi_t(h^T)\pi_t(\phi))(dY_t-\pi_t(h)dt)
\end{equation}
\end{theorem}
\begin{theorem}(ref \cite{Pl})
Let ${X_t}, {t\ge 0}$ is a semimartingale and ${Y_t}, {t\ge 0}$ is a counting process such that
\begin{equation}
X_t = X_0 + \int_0^t H_sds + Z_t
\end{equation}
where $Z_t$ is a $\mathscr{F}_t$-martingale, $H_t$ is a bounded $\mathscr{F}_t$-progressive process, $E[sup_{s \leq {t}}|X_s|] < \infty$ for every $t \geq {0}$ and $X_0$ is a square integrable randaom variable.
\begin{equation}
Y_t = Y_0 + \int_0^t h_sds + M_t
\end{equation}
where $M_t$ is a $\mathscr{F}_t$-martingale with mean 0 and $\sigma(M_u - M_t;u \ge {t})$ is independent of $\sigma(Y_u,h_u; u \ge {t})$ and $h_t=h(X_t)$ is a positive bounded $\mathscr{F}_t$-progressive process such that $E\int_0^t h_s^2ds < {\infty}, \forall {t}$.
Then the filtering equation is given by
\begin{equation}
\pi(X_t) = \pi(X_0) + \int_0^t \pi(X_s)ds + \int_0^t [\pi(h_s)]^{-1}[\pi(X_{s^-}-h_s)-\pi(X_{s^-})\pi(h_s)]dm_s,
\end{equation} where
\begin{equation}
m_t = Y_t \int_0^t \pi(h_s)ds
\end{equation}
\end{theorem}
Now we are ready to derive the main result of classical versions for quantum filters.
\begin{theorem} A classical version exists for the quantum filtering problem defined by the equations $\eqref{sig}, \eqref{hed}, and \eqref{phc}$.
\begin{proof}
Let us apply spectral theorem on the commutative von Neumann algebra $(\mathscr{Y}_t, \mathbb{P})$ of the observation process (Ref \cite{Sa} Proposition 1.18.1). The abelian $\mathscr{Y}_t$ is *-isomorphic to a W* algebra $L^\infty(\Omega,\mathcal{F},\mu)$. Let us outline some of the main steps of the proof.

The linear functional $\mathbb{P}$ defines a Radon measure $\mu$ such that
\begin{equation}
\mathbb{P}(a) = \int_{\Omega}a(t)d\mu   \qquad  \forall a \in \mathscr{Y}
\end{equation}
The mapping $\Phi:a \rightarrow a(t)$ is a *-homomorphism such that
\begin{equation}
\mathbb{P}(ab) = \int_{\Omega}a(t)b(t)d\mu \qquad  \forall a,b \in \mathscr{Y}
\end{equation}
Therefore the target random variables a(t) are square integrable.

We then apply Radon-Nikodyn theorem to get a probability measure $\phi$ that is absolutely continuous with respect to $\mu$. In the vacuum state the quantum Martingale component of the observation processes are mapped to classical square integrable Martingales, a continuous Martingale in the case of Homodyne detection and a c$\grave{a}$dl$\grave{a}$g Martingale in the case of photon counting, with respect to the natural filtration, that can be represented as integral with respect to Brownian motions.

On the same probability space using Thoerem 1 we can replace the non-commutative process $Z_t$ with its classical version $\tilde{Z}_t$ which is a Brownian motion. \\
So the classical version of the signal and observation processes become,
\begin{equation}
dj_t(\tilde{X}) = \mathrm{E}_{\phi}(\iota(\mathscr{L}_{L,H}(\tilde{X}))) + d\tilde{Z}_t
\end{equation}
\begin{equation}
d\tilde{\mathscr{Y}}_t^W = j_t(L + L^*)dt + d\tilde{V}_t
\end{equation}
\begin{equation}
d\tilde{\mathscr{Y}}_t^\Lambda = j_t(L^*L)dt + d\tilde{U}_t
\end{equation}
where $\tilde{Z}, \tilde{V}, \tilde{U}$ are all independent Brownian motions.
We get by applying Theorem 2 for the Homodyne detection and Theorem 3 for photon counting respectively,
\begin{equation}
d\pi_t(\tilde{X}) = \pi_t(\mathscr{L}_{L,H})dt + (\pi_t(L^*\tilde{X} + \tilde{X}L) - \pi_t(L^* + L)\pi_t(\tilde{X}))(d\tilde{\mathscr{Y}_t} - \pi_t(L^* + L)dt)
\end{equation}
\begin{equation}
d\pi_t(\tilde{X}) = \pi_t(\mathscr{L}_{L,H})dt + (\frac{\pi_t(L^*\tilde{X}L)}{\pi_t(L^*L)} - \pi_t(\tilde{X}))(d\tilde{\mathscr{Y}_t} - \pi_t(L^*L)dt)
\end{equation}
\end{proof}
\end{theorem}

\end{document}